\documentclass{amsart}

\usepackage{amsmath}

\usepackage{paralist}

\newtheorem{theorem}{Theorem}[section]

\newtheorem{lemma}[theorem]{Lemma}
\newtheorem{proposition}[theorem]{Proposition}

\theoremstyle{definition}

\newtheorem{remark}[theorem]{Remark}

\newtheorem{example}[theorem]{Example}

\newcommand{\can}{\overline{\phantom{x}}}

\renewcommand{\bar}{\overline}

\newcommand{\cN}{\mathcal{N}}
\newcommand{\cC}{\mathcal{C}}
\newcommand{\cO}{\mathcal{O}}
\newcommand{\x}{\times}
\newcommand{\End}{\mathrm{End}}
\newcommand{\Mat}{\mathrm{Mat}}
\newcommand{\Cay}{\mathrm{Cay}}
\newcommand{\Aut}{\mathrm{Aut}}

\newcommand{\Q}{\mathbb{Q}}
\newcommand{\Z}{\mathbb{Z}}
\newcommand{\N}{\mathbb{N}}
\newcommand{\R}{\mathbb{R}}
\newcommand{\C}{\mathbb{C}}

\allowdisplaybreaks

\title[Space-time block codes from nonassociative division algebras]
      {Space-time block codes from nonassociative division algebras}

\author[Susanne Pumpl\"un and Thomas Unger]{}


\date{09.06.2011}

\begin{document}
\maketitle

\centerline{\scshape Susanne Pumpl\"un}
\medskip
{\footnotesize
 \centerline{School of Mathematical Sciences}
  \centerline{University of Nottingham}
   \centerline{University Park,
  Nottingham NG7 2RD, United Kingdom}
   }
   
\bigskip

\centerline{\scshape Thomas Unger}
\medskip
{\footnotesize
 \centerline{School of Mathematical Sciences}
  \centerline{University College Dublin}
   \centerline{Belfield, Dublin 4, Ireland}
   }

\begin{abstract}
Associative division algebras are a rich source of fully diverse space-time block codes (STBCs). In this paper the systematic construction of fully diverse  STBCs from nonassociative algebras is discussed.  As examples, families of fully diverse $2\times 2$, $2\times 4$ multiblock   and $4\x 4$ STBCs
are designed, employing nonassociative quaternion division algebras. 
\end{abstract}

\section{Introduction}
Space-time coding is used for reliable high rate transmission
over wireless digital channels with multiple antennas at both the transmitter and receiver ends.
From a mathematical point of view,  a space-time block code (STBC) consists of a
family of  matrices with complex entries (the codebook)  that satisfies a number of properties which determine how well the code performs.

The first aim is to find \emph{fully diverse} codebooks, where the difference of any two 
code words   has full rank. Once a fully diverse codebook is found it is then further optimized to satisfy additional design criteria (see Section~\ref{sec6}).

Using central simple associative division algebras to build
space-time block codes allows for a systematic code design
(see for instance \cite{O-R-B-V}, \cite{E-S-K}, \cite{S-R-S}, \cite{H-L-R-V}, \cite{Be-Og1}, \cite{Be-Og2}, \cite{Be-Og3} and the excellent survey \cite{Seth}).

Most of the existing  codes  are built from cyclic division algebras
over $F=\mathbb{Q}(i)$ or $F= \mathbb{Q}(\zeta_3)$
with $\zeta_3=e^{2\pi i /3}$ a third root of unity. These fields are used for
 the transmission of QAM or HEX constellations, respectively.

There are two  ways to embed an associative division algebra into a matrix algebra in order to obtain a
 codebook: the left regular
representation of the algebra and the representation over some maximal subfield. For instance, a real $4\times 4$ orthogonal design is obtained
by the left regular representation of the real quaternions $\mathbb{H}=(-1,-1)_\mathbb{R}$
(see for instance \cite[p.~1458]{T-J-C1}), whereas the  Alamouti code \cite{Al} uses the representation of $\mathbb{H}$ over its maximal
subfield $\mathbb{C}$.

In \cite[p.~2608]{S-R-S}, the authors note that ``the Alamouti code is the only rate-one STBC
which is full rank over any finite subset of $\mathbb{C}$, which is due to the fact that the set of quaternions $\mathbb{H}$
is the only division algebra which has the entire complexes as its maximal subfield.''
There are however other, \emph{nonassociative}, division algebras over $\mathbb{R}$ of dimension 4 which contain $\mathbb{C}$
as a subfield and which yield  STBCs that are full rank over any finite subset of $\mathbb{C}$:
the  \emph{nonassociative quaternion division algebras} over $\mathbb{R}$ which were classified in \cite{Al-H-K} and which we will employ here.

In this paper we show how  nonassociative division algebras can be
used  to systematically construct fully diverse linear STBCs.
If a nonassociative algebra $A$ behaves well enough,
 one can also obtain fully diverse families of matrices using subfields of $A$.
We  use 4-dimensional nonassociative quaternion division algebras 
to construct new
examples of  fully diverse $2\times 2$ and $4\x 4$ space-time block codes and of $2\times 4$ multiblock space-time  codes (cf. \cite{L-M-M}). We also investigate when these codes satisfy the non-vanishing determinant property.

The paper is organized as follows: in Sections~\ref{sec2} and \ref{sec3}
we present nonassociative algebras and the Cayley-Dickson doubling process, respectively. In Section~\ref{sec4} we define nonassociative quaternion division algebras (constructed via a
generalized Cayley-Dickson doubling process). In Section~\ref{sec5} we explain the general framework for obtaining fully diverse STBCs from nonassociative division algebras. 
In Section~\ref{sec6} we list the design criteria used in the construction of STBCs.
In Section~\ref{sec7} we look in more detail at the construction of fully diverse $2\x 2$ STBCs from nonassociative quaternion algebras. We also discuss the non-vanishing determinant  and 
information lossless  properties. This is followed by many examples. 
In Sections~\ref{sec8} and \ref{sec9} we discuss the construction of fully diverse $2\x 4$ multiblock codes and $4\x 4$ codes from nonassociative quaternion algebras, respectively. In Appendix~\ref{AppA} we collect results from algebraic number theory that are needed in this paper.

\section{Nonassociative algebras}\label{sec2}

Let $F$ be a field and let $A$ be a finite-dimensional $F$-vector space.
We call $A$ an \emph{algebra} over $F$ if there exists an
$F$-bilinear map $A\times A\to A$, $(x,y) \mapsto x \cdot y$ (also denoted simply by juxtaposition $xy$),
called  \emph{multiplication}, on $A$.  This definition does not imply that the algebra is associative; we only have
$c(xy)=(cx)y=x(cy)$ for all $c\in F$, $x,y\in A$. Hence we also call such an algebra
a \emph{nonassociative algebra}, in the sense that it is not necessarily associative.
A (nonassociative) algebra $A$ is called \emph{unital} if there is
an element in $A$ (which can be shown to be uniquely determined), denoted by 1, such that $1x=x1=x$ for all $x\in A$.
We will only consider unital nonassociative algebras.

For a nonassociative $F$-algebra $A$, associativity in $A$ is measured by the {\it
associator} 
\[[x, y, z]  = (xy) z - x (yz).\]
The \emph{nucleus} of $A$ is
defined as
\[\cN(A)  = \{ x \in A \, \vert \, [x, A, A] = [A, x, A] = [A,
A, x] = 0 \}.\]
The nucleus is an associative subalgebra of $A$ (it may be
zero), and $x(yz) = (xy) z$ whenever one of the elements $x, y, z$ is in
$\cN(A)$. In other words, the nucleus of the algebra $A$ contains all the elements of $A$
which associate with every other two elements in $A$. Suppose $K\subset A$ is a subfield of $A$.
An $F$-algebra $A$ is \emph{$K$-associative} if $K$ is contained in the nucleus $\cN(A)$.
The \emph{left nucleus} of $A$ is defined as
$$\cN_\ell(A) = \{ x \in A \, \vert \, [x, A, A]  = 0 \},$$
 the \emph{middle nucleus} of $A$ is
defined as
$$\cN_m(A) = \{ x \in A \, \vert \, [A, x, A]  = 0 \}$$
 and  the \emph{right nucleus} of $A$ is
defined as
$$\cN_r(A) = \{ x \in A \, \vert \, [A,A, x]  = 0 \}.$$
Their intersection is the nucleus $\cN(A)$.

 A nonassociative algebra $A$ is called a \emph{division algebra} if for any $a\in A$, $a\not=0$,
the left multiplication  with $a$, $\lambda_a(x)=ax$,  and the right multiplication with $a$, $\rho_a(x)=xa$, are bijective. The algebra
$A$ is a division algebra if and only if $A$ has no zero divisors  \cite[pp. 15, 16]{Sch}.
Note that if the $F$-algebra $A$ is associative and finite-dimensional as an $F$-vector space, this
definition of division algebra coincides with the usual one for associative algebras.

\section{The Cayley-Dickson doubling process}\label{sec3}

The Cayley-Dickson doubling process is a well-known way to construct a new
algebra with involution from a given algebra with involution.
It can be motivated by the observation that the complex numbers can be viewed as pairs of real numbers with componentwise
addition  and a suitably defined multiplication:

\begin{example}
We define a multiplication on  $\mathbb{R}\x \mathbb{R}$ via
$$(u,v)(u',v'):=(uu' -v'v,uv'+u'v),$$
for $u,v,u',v'\in \mathbb{R}$. The unit element for this multiplication is $(1,0)$.
Let $i=(0,1)$. Then $i^2=(-1,0)$. We can now write the pair
 $(u,v)$ as $(u,v)=(u,0)+(0,1)(v,0)$ and identify it with
 the element $u+iv\in \mathbb{R}\oplus  i\mathbb{R}$.
In this way  we obtain the complex numbers 
$$\mathbb{C}=\mathbb{R}\oplus  i\mathbb{R}.$$
For $x=u+iv$, $y=u'+iv'$ with $u,v,u',v'\in \mathbb{R}$, we have
$x y=(uu' -v'v)+i(v'u+vu')$.

Let $\can$ denote  complex conjugation, given by $\bar{x}=u-iv$ for $x=u+iv$. Then we can also write $\overline{(u,v)}=(u,-v)$.

The above process can be repeated with $\C$ instead of $\R$: define a multiplication on  $\mathbb{C}\x \mathbb{C}$ via
$$(u,v)(u',v'):=(uu' -v'\bar{v}, \bar{u}v'+u'v),$$
for $u,v,u',v'\in \mathbb{C}$. The unit element for this multiplication is $(1,0)$. 
Let $j=(0,1)\in \C\x \C$. Then $j^2=(-1,0)$.
We identify the element $(u,v)\in \C\x\C$ with $u+jv\in \C\oplus j\C $. In this way we obtain the Hamilton quaternions
\[\mathbb{H}=\C \oplus j\C .\]
We define quaternion conjugation (again denoted $\can$) via
\[\bar{(u,v)}=(\bar{u}, -v).\]
\end{example}

Another iteration of this process, this time starting with $\mathbb{H}$, results in the (Cayley-Graves) octonion algebra $\mathbb{O}$.

\begin{remark}
In 1958 it was shown that finite-dimensional real division algebras can only
have dimension 1, 2, 4 or 8 (cf. \cite{Bo-M}).
In addition to the well-known algebras
$\mathbb{R}$, $\mathbb{C}$, $\mathbb{H}$ and $\mathbb{O}$, there exist other finite-dimensional real division algebras. The algebras
$\mathbb{R}$, $\mathbb{C}$, $\mathbb{H}$ and $\mathbb{O}$ are just  the alternative ones (see \cite[p.~48]{Sch}).
Indeed,  a complete classification 
of these algebras is still far away. Only certain subclasses are understood thus far.
Moreover, the restriction on the dimension only holds for real and real closed fields (see \cite{D-D-H}).
Over number fields there exist also higher dimensional division algebras as
well as division algebras which do not appear over the
real numbers.
\end{remark}

The previous construction of the Hamilton quaternions as a double of the complex numbers serves as a motivating example for obtaining generalized
 (associative) quaternion algebras as doubles,  as we will do now.

Let $F$ be a field. Let $K$ be  a separable quadratic field extension of $F$ with non-trivial Galois automorphism $\sigma:K\to K$.
Let  $b\in F^\times:=F\setminus \{0\}$. Then the 4-dimensional $F$-vector space
 $K\x K$ can be made into a new unital associative (but not commutative) algebra over $F$ via the multiplication
$$(u,v)(u',v'):=(uu'+b v'\sigma(v),\sigma(u)v'+u'v)$$
for $u,u',v,v'\in K$. The unit element is given by $(1,0)$. The automorphism $\sigma$ induces an involution $\can$ on $K\x K$ as follows:
\[\bar{(u,v)}:= (\sigma(u), -v).\]
Let $j=(0,1)$. Then $j^2=(b,0)$. We identify $(u,v)\in K\x K$ with $u+jv$ in $K\oplus jK $.
The algebra $K\oplus jK$ is called the \emph{Cayley-Dickson double of $K$ $($with scalar $b$$)$} and denoted by $\Cay(K,b)$ (cf. \cite{A2}).

The Cayley-Dickson double of $K$ yields a quaternion algebra over $F$. If $F$ has characteristic not 2 and
$K=F(\sqrt{a})=F(i)$, $\sigma: \sqrt{a} \mapsto -\sqrt{a}$, we have
$$\Cay(K,b)\cong (a,b)_{F}.$$
The standard basis $\{1,i,j,k\}$ of the quaternion algebra $(a,b)_F$ satisfies $i^2=a$, $j^2=b$,
$k=ij$ and $ij=-ji$. 

The Cayley-Dickson doubling process depends on the scalar $b$ only up to an invertible square, i.e.
$$\Cay(K,b)\cong \Cay(K,bd^2)$$
for every $d\in F^\times$. The algebra $\Cay(K,b)$ is a division algebra if and only if
$b\not \in N_{K/F}(K^\times)$, where $N_{K/F}$ is the norm of the
field extension $K/F$.

The quadratic \emph{norm} $N_A:A\to F$ of the algebra $A=\Cay(K,b)$ is given by
$$N_A(u+jv)=N_{K/F}(u)-b N_{K/F}(v)$$
for $u,v\in K$.
If $F$ has characteristic not 2, a straightforward computation shows that 
$$N_A(x)=x\bar{x}=\bar{x}x=q_0^2-aq_1^2-bq_2^2+abq_3^2$$
for $x=q_0+iq_1+jq_2+jiq_3$, $q_i\in F$.

\begin{example}
Let $\mathbb{H}=(-1,-1)_\mathbb{R}$ denote Hamilton's quaternion algebra.
This is just the algebra $\Cay(\mathbb{C},-1)$, as already established above.
\end{example}

\begin{example} The algebra
$(5,i)_{\mathbb{Q}(i)}=\Cay(K,i)$ with $K=\mathbb{Q}(i,\sqrt{5})$ is the quaternion algebra used
in the construction of the
Golden code \cite{Og}. This algebra is isomorphic to the cyclic algebra $(K/F,\sigma, i)$ where $\sigma:\sqrt{5}\mapsto-\sqrt{5}$.
Since $K=\mathbb{Q}(i,\sqrt{5})\cong\mathbb{Q}(i, \theta)$, where $\theta=
\frac{1+\sqrt{5}}{2}$ is the golden number, we also have that
$\Cay(\mathbb{Q}(i,\theta),i)\cong (5,i)_{\mathbb{Q}(i)}.$
\end{example}

\begin{remark} The Cayley-Dickson doubling process can be iterated: the quaternion algebras double to \emph{octonion algebras}, which in turn double to \emph{sedenion algebras}. Continuing the doubling process results in successive \emph{generalized Cayley-Dickson algebras}. 
\end{remark}

\section{Nonassociative quaternion division algebras}\label{sec4}

Let $F$ be a field of characteristic not 2. 
Let $K$ be a quadratic field extension of $F$ with non-trivial Galois automorphism
$\sigma$ and let $b\in K\setminus F$. We define an algebra structure on
the $F$-vector space $K\x K$ via the multiplication
$$(u,v)(u',v'):=(uu'+b v'\sigma(v), \sigma(u)v'+u'v)$$
for $u,u',v,v'\in K$. The multiplication is thus defined just
as for  quaternion algebras with the exception that we require the scalar $b$ to lie outside of $F$. We denote the algebra again by 
$\Cay(K,b)$. Its  unit element is $(1,0)$.

Since  $b\in K\setminus F$, the multiplication of  $\Cay(K,b)$ is not associative
anymore. It is not even third power-associative, meaning that in general
$(x^2)x\not=x(x^2)$. The algebra
$\Cay(K,b)$ with $b\in K$ and not in $F$ is called a \emph{nonassociative quaternion algebra} over $F$.

\begin{remark}\label{rem4.1}
 Let $K=F(\sqrt{a})=F(i)$ be a quadratic field extension and let $b\in K\setminus F$. Let $A=\Cay(K,b)$ be a nonassociative quaternion
 division algebra. Put
 $$j=(0,1)\in \Cay(K,b).$$
Then $A$ has $F$-basis $\{1,i,j,ji\}$ such that $i^2=a$, $j^2=b$  and
 $xj=j\sigma(x)$
 for all $x\in K$ (so in particular $ij=-ji$). 
\end{remark}

\begin{theorem}[{\cite{As-Pu}} or {\cite{W}}]\label{thm4}
The nonassociative quaternion algebra $\Cay (K, b) $
has nucleus $K$ and is a division algebra over $F$.
\end{theorem}

Thus products involving a factor from $K$ are still associative. Furthermore, nonassociative quaternion algebras are always division algebras, which is not the case for the usual associative quaternion algebras.

\begin{remark}\label{rem4.3}
Let $F=\mathbb{R}$ and $b,b'\in\C$. Let
$b=p+iq$. Then $\Cay (\mathbb{C}, b) \cong \Cay (\mathbb{C}, b')$
if and only if $b'=t(p\pm iq)$ for some positive $t\in \mathbb{R}$  \cite[Thm. 14]{Al-H-K}.

Over $\mathbb{Q}$, we can easily find non-isomorphic nonassociative quaternion division algebras:
it was observed in \cite{W} that two nonassociative quaternion  algebras $\Cay(K,b)$ and $\Cay(L,c)$ can only be isomorphic if $L\cong K$.
Moreover,
$$\Cay(K,b)\cong \Cay(K,c) \text{ iff } g(b)=d\sigma(d)c$$
for some automorphism $g\in \Aut(K)$ and some non-zero $d\in K$.
\end{remark}

 Nonassociative quaternion algebras provided early examples of real nonassociative division algebras which
were neither power-associative nor quadratic.
They were investigated for the first time by Dickson \cite{D} in 1935 and by Albert \cite{A1} in 1948.
In 1987, Waterhouse \cite{W} completely classified these algebras over a field of characteristic not 2.

The only division algebras which appear in the classification of 4-dimensional $K$-associative algebras, cf. \cite{Al-H-K} or \cite{W}, are the generalized quaternion division algebras and the
 nonassociative quaternion division algebras over $F$.

\section{STBCs from nonassociative division algebras: the general setup}\label{sec5}

The general setup for constructing a fully diverse  STBC from an associative division algebra $A$ is simple: associate to each nonzero element $x\in A$ a square matrix $X$ over a fixed subfield of $A$ (normally the base field, via the left regular representation, or a maximal extension of the base field). The difference of any two such matrices $X-X'$ (with $X\not=X'$) will then always be invertible.
This procedure can be adapted to work in the nonassociative case as well, as will be explained in this section.

\subsection{The left regular representation}
Let $A$ be a nonassociative division algebra over $F$ of dimension $n$ as an $F$-vector space.
Let $a$ be any element in $A$. The left multiplication $\lambda_a:A\to A$ determined by $a$ is defined
by $x\mapsto ax$ for all $x\in A$. The operator $\lambda_a$ is  linear and the set $\{\lambda_a\,|\, a\in A\}$ is a  subspace of the
associative algebra $\End_F(A)$, the algebra of $F$-linear transformations on $A$.
Consider the left regular representation
$$\lambda: A \to \End_F(A), a\mapsto \lambda_a.$$
If $\lambda_a=\lambda_b$ then $ax=bx$ for all $x\in A$, hence $(a-b)x=0$ for all $x$, which yields $a=b$
($A$ is a finite-dimensional division algebra and as such does not have zero divisors) and we have an injection.

After a choice of $F$-basis for $A$, we can embed $\End_F(A)$ into the algebra ${\rm Mat}_n(F)$ of $n\x n$-matrices with entries from $F$, where $n=\dim_F(A)$.
In this way we get an embedding $\lambda:A\hookrightarrow {\rm Mat}_n(F)$ of vector spaces.

Contrary to the situation for associative division algebras, this only
embeds the vector space $A$ into the vector space ${\rm Mat}_n(F)$; the algebra structure of $A$ is disregarded here.

Nonetheless, all non-zero elements of $A$ are invertible, hence all $\lambda_a$ with $a\not=0$ are bijective and so all
non-zero matrices in $\lambda(A)$ have non-zero determinant. Now
$X\pm Y\in \lambda(A)$ for all $X,Y\in \lambda(A)$.
Thus $\lambda(A)$ constitutes a linear codebook which in addition is fully diverse,
since the rank of the difference of two distinct codewords is maximal.

\subsection{Representation over a maximal subfield}
For coding purposes, an associative division algebra $A$ is often  considered as a vector space over some
 subfield $K$ of the algebra $A$. Usually $K$ is maximal with respect to inclusion.
Given a nonassociative $F$-algebra $A$ with a maximal subfield $K$, this is not always possible
because of the nonexistence of the associative law. So what are the minimum requirements on a nonassociative algebra in order
to have such a representation?

Let $K$ be a subfield of the $F$-algebra $A$.
We need $A$ to be a right $K$-vector space, i.e. we need
$$x(cd)=(xc)d \text{ for all }x\in A, c,d\in K.$$
This is satisfied for instance if $K \subset \cN_r(A)$ or if $ K \subset \cN_m(A) $.

We also need that left multiplication $\lambda_a$ is a linear endomorphism of the right $K$-vector space
 $A$, i.e. that $(\alpha a)x=\alpha ( ax)$ for all $\alpha\in K$, $a,x\in A$, which is equivalent to
 $K\subset \cN_\ell(A)$. 
 Then
$$\lambda_{\alpha a}(x)=(\alpha a)x=\alpha ( ax)=\alpha \lambda_a(x)$$
for all $a,x\in A$, $\alpha\in K$ and $\lambda_a\in \End_K(A)$, so $\lambda: A \hookrightarrow \End_K(A), a\mapsto \lambda_a.$

Thus,  let $K$ be a subfield of $A$, maximal with respect to inclusion and assume that $K \subset \cN_r(A)\cap \cN_\ell(A)$ or $ K \subset \cN_m(A) \cap \cN_\ell(A)$.
 Consider $A$ as a right $K$-vector space.
After a choice of  $K$-basis for $A$, we can embed $\End_K(A)$ into the vector space
${\rm Mat}_r(K)$ where $r=\dim_K (A)$. In this way we get an embedding
$$\lambda:A\hookrightarrow {\rm Mat}_r(K)$$
 of vector spaces. Obviously, we have $X\pm Y\in \lambda(A)$ for all $X,Y\in \lambda(A)$.
Thus $\lambda(A)$ constitutes a linear codebook.

\begin{remark}
If we want to consider $A$
as a \emph{left} $K$-vector space, we 
require $K\subset \cN_\ell(A) $ and $K\subset \cN_m(A)$ and  adjust the above construction accordingly. 
\end{remark}

More generally one can do the following:
let $D$ be a subalgebra of $A$,  assume that
$D\subset \cN_r(A)\cap \cN_\ell(A)$ or that
$D \subset \cN_m(A) \cap \cN_\ell(A)$
 and suppose $A$ can be viewed as a free right $D$-module of rank $r$.
After a choice of a $D$-basis for $A$, we can embed the right $D$-module $\End_D(A)$ into the vector space
${\rm Mat}_r(D)$. In this way we get an embedding
$\lambda:A\to {\rm Mat}_r(D)$
 of $D$-modules. Obviously, we have $X\pm Y\in \lambda(A)$ for all $X,Y\in \lambda(A)$.
Thus $\lambda(A)$ is a linear codebook.

\begin{remark}
It is not known whether there exist 8-dimensional real division algebras with some (left, middle or right)
nucleus isomorphic to $\mathbb{H}$. The fact that there are no 8-dimensional real division algebras with
two associative nuclei (left, middle or right) isomorphic to $\mathbb{H}$ suggests a negative answer \cite[Proposition~3]{J-P}.
This need not be the case over other base fields, however.
\end{remark}
\medskip

\emph{In the remainder of this paper, all fields are assumed to be algebraic number fields unless stated otherwise.}

\section{STBC Design criteria}\label{sec6}

Let $\cC\subset \Mat_n(\C)$ be a space-time block code. In order for $\cC$ to perform well, it should satisfy property~(1) below (as remarked before) and as many of the other properties as possible.

 \begin{enumerate}[(1)]
 
\item It is \emph{fully diverse}: $\det(X-X')\not=0$ for all matrices $X\not=X',$ $X,X'\in \mathcal{C} $.
\item It has \emph{full rate}, which means that the $n^2$ degrees of freedom are used to transmit $n^2$ information symbols.
\item It has \emph{non-vanishing determinant}: the minimum determinant
of the code,
$$\delta(\mathcal{C})=\inf_{X'\not=X''\in \mathcal{C}}|\det(X'-X'')|^2,$$
 is bounded below by a constant even if the codebook $\mathcal{C}$   is infinite. 
\item It has \emph{cubic shaping}: each layer of a codeword is of the form $Rv$, where $R$ is a
unitary matrix and $v$ is a vector containing the information symbols.  As a consequence it is \emph{information lossless}.

\item It induces \emph{uniform average energy per antenna}: the $i$th antenna will transmit the $i$th row of the codeword;
on average, the norms of the rows should be equal in order to have a balanced repartition of
the energy at the transmitter.
 \end{enumerate}

These properties, originally considered for codes based on associative division algebras, also make sense in the nonassociative case.
Codes that satisfy all the properties above are called \emph{perfect codes}. We refer to \cite{O-R-B-V} for more details. The Golden Code \cite{B-R-V} is the best performing $2\x 2$ perfect STBC, cf. \cite{Og}.

\section{$2\times 2$ Codebooks from nonassociative quaternion division algebras}
\label{sec7}

STBCs based on associative quaternion algebras seem to have been considered explicitly for the first time in \cite{BR}. See also \cite{U-M} for more details. In this section we  look at the construction of STBCs based on nonassociative quaternion algebras.

Roughly speaking, constructing a nonassociative quaternion division
algebra boils down to choosing the nonzero scalar $b$ in the quadratic field extension $K=F(\sqrt{a})$ of the base field
$F$, and \emph{not} in $F$ itself. In contrast, $b$ \emph{is} chosen in $F$ in
the construction of a classical generalized quaternion
 algebra $(a,b)_F$ over $F$.  This usually gives us more
freedom of choice for $b$, despite the restriction that we will still have to require $|b|^2=1$ in order to get a
balanced repartition of the energy at the transmitter.
The choice of $F=\mathbb{Q}(i)$ allows us to transmit QAM constellations.

\subsection{Fully diverse codebook construction} \label{sec6.1}
Let $K$ be a quadratic field extension of $F$ with non-trivial Galois automorphism $\sigma$.
Let $A=\Cay(K,b)$  be a nonassociative quaternion division algebra over $F$ (so $b\in K\setminus F$) with $F$-basis $\{1,i,j,ji\}$ (see Remark~\ref{rem4.1}).

The algebra  $A$ is $K$-associative by Theorem~\ref{thm4}, hence we can consider $A$ as a right vector space over the subfield $K$ of $A$. The field $K$ is maximal with respect to inclusion. For $x\in A$, the left multiplication
 $\lambda_x:A\to A$, $a\mapsto xa$, is a $K$-linear endomorphism of the right $K$-vector space $A$.
Therefore $\lambda_x\in \End_K(A)$ and we get an injective $K$-linear map
$$\lambda: A \hookrightarrow \End_K(A), x\mapsto \lambda_x.$$
Consider the $K$-basis $\{1,j\}$ of $A$. Then $\End_K(A)\cong \Mat_2(K)$ as vector spaces and
we get an embedding $\lambda:A\hookrightarrow {\rm Mat}_2(K)$ of vector spaces, which sends $x\in A$ to the matrix of $\lambda_x$ with respect to the basis
$\{1,j\}$.

\begin{lemma}\label{lem6.1}
\[\lambda(A)\cong \Biggl\{\left [\begin {array}{cc}
x_0 &  b\sigma(x_1)\\
x_1 & \sigma(x_0)
\end {array}\right ] \,\Bigg\vert\, x_0,x_1\in K \Biggr\}
 \]
\end{lemma}

\begin{proof}
Let $x=x_0+jx_1\in A$ with $x_0, x_1 \in K$. Then
\begin{align*}
\lambda_x(1)&=x_0+jx_1,\\
\lambda_x(j)&=x_0j+jx_1j=j\sigma(x_0)+j^2\sigma(x_1)=b\sigma(x_1)+j\sigma(x_0)
\end{align*}
by the rules in Remark~\ref{rem4.1}.
\end{proof}

\begin{lemma}
For any
\[0\not=X= \left[\begin {array}{cc}
x_0 &  b\sigma(x_1)\\
x_1 & \sigma(x_0)\\
\end {array}\right ]
 \]
with $x_i\in K$ for $i=1,2$,  we have
$$\det(X)=N_{K/F}(x_0)-b N_{K/F}(x_1)\not=0.$$
\end{lemma}

\begin{proof} For $b=u+\sqrt{a}v\in K$, $u,v\in F,$ $v\not=0$, we compute
$$\det(X)=x_0\sigma(x_0)- b\sigma(x_1)x_1=(N_{K/F}(x_0)-u N_{K/F}(x_1))-\sqrt{a}v  N_{K/F}(x_1).$$
Since $(x_0,x_1)\not=(0,0)$ and since $N_{K/F}(x)=0$ iff $x=0$, we get
$\det(X)\not=0$.
\end{proof}

Since $X\pm Y\in \lambda(A)$ for all $X,Y\in \lambda(A)$, the difference of any two distinct elements in $\lambda(A)$ will have non-zero determinant.
 Therefore the (infinite linear) codebook built on $A$, $\cC:=\lambda(A)$,
is fully diverse.

\subsection{Non-vanishing determinant} \label{S6.2}

We closely follow the approach in \cite[\S17]{Be-Og2}.
The minimum determinant of $\cC$ determines the coding gain and is defined as
$$\delta(\mathcal{C})=\inf_{X'\not=X''\in \mathcal{C}}|\det(X'-X'')|^2.$$
The discussion in
 \cite[p.~73]{Be-Og2} can easily be adapted to the more general set-up of  nonassociative
algebras. Since the codebook $\cC$ is linear (it is based on an algebra) we have
$$\delta(\mathcal{C})=\inf_{0\not=X\in \mathcal{C}}|\det(X)|^2.$$
Let us compute the minimum determinant  of the codebook
\[\mathcal{C}=\Biggl\{\left [\begin {array}{cc}
c+d\sqrt{a} &  b(e-f\sqrt{a})\\
e+f\sqrt{a} & c-d\sqrt{a}
\end {array}\right ]\,\Bigg|\,c,d,e,f\in F \Biggr\},
 \]
obtained from the nonassociative quaternion division algebra $A=\Cay(K,b)$ with $K=F(\sqrt{a})$ and $b\in K\setminus F$ in \S\ref{sec6.1}.
We obtain
$$\delta(\mathcal{C})=\inf_{c,d,e,f\in F}|N_{K/F}(c+d\sqrt{a})-b
N_{K/F}(e+f\sqrt{a})|^2$$
with the infimum taken over all $(c,d,e,f)\not=(0,0,0,0)$, or equivalently
$$\delta(\mathcal{C})=\inf_{c,d,e,f\in F}|c^2-ad^2-be^2+abf^2|^2.$$
Thus
 $$\delta(\mathcal{C})\in K\cap \mathbb{R}^+.$$
Since $A$ is a division algebra, $\delta(\cC)\not=0$. If the code $\cC$ is finite, i.e.
 if the information symbols $c,d,e,f$ belong to a finite constellation in $F$, then $\delta(\cC)$ is bounded below by a constant. If the constellation size increases however, $\delta(\mathcal{C})$ can get arbitrarily
close to zero (e.g. let $(c,d,e,f)=(\frac{1}{n},0,0,0)$; as $n$ increases, $\delta(\cC)$ will approach zero). This will also be the case for infinite codes.

Codes whose  minimum determinant is bounded below by a constant which is independent of the size of the constellation from which the information symbols are chosen are said to satisfy  the \emph{non-vanishing determinant} (NVD) property, cf. Section~\ref{sec6}.

For associative division algebras over a number field $F$ and with maximal subfield $K$ infinite codes that satisfy
 the NVD property can often be obtained by restricting the entries in the codebook to the ring of integers $\mathcal{O}_K$. If $F=\Q$ or $F$ is  quadratic imaginary, then the resulting code will still be infinite, and its minimum determinant is guaranteed to be bounded
 away from zero, cf. \cite[Cor.~17.8]{Be-Og2}.
 
 Let us look at what happens for a code $\cC$, based on a nonassociative quaternion division algebra.

\begin{proposition}\label{prop6.2}
 Let $F$ be a number field and let $K=F(\sqrt{a})$ for some nonzero square-free
$a\in \mathcal{O}_F$. Let $b\in K\setminus F$. 
Let $\cC=\lambda(\Cay(K,b))$ and let $\mathcal{C}_{\mathcal{O}_K}$ denote the code  obtained from $\cC$
by restricting the elements of $K$ to elements of $\mathcal{O}_K$.  
Then there exists a constant $c>0$ such that 
\[\delta(\mathcal{C}_{\mathcal{O}_K})\in \frac{1}{c} \cO_K\cap \mathbb{R}^+.\]
If $K$ is quadratic imaginary, then there exists an integer $d>0$ such that
\[\delta(\mathcal{C}_{\mathcal{O}_K})\geq \frac{1}{d}\]
$($and so $\mathcal{C}_{\mathcal{O}_K}$ satisfies the NVD property$)$, otherwise $\delta(\mathcal{C}_{\mathcal{O}_K})$ can become arbitrarily small.
\end{proposition}
 
 \begin{proof} Write $b$ as a fraction $b=\frac{b_n}{b_d}$ with $b_n,b_d\in \mathcal{O}_K$
(not necessarily unique)
 and $b_d\not=0$. A codeword of $\mathcal{C}_{\mathcal{O}_K}$ is of the form 
\[ \left[\begin {array}{cc}
u &  b\sigma(v)\\
v & \sigma(u)\\
\end {array}\right ]
 \]
with $u,v\in \cO_K$. Thus we have
\begin{align*}
\delta(\mathcal{C}_{\mathcal{O}_K})&=\inf_{u,v\in \mathcal{O}_K}|N_{K/F}(u)-b N_{K/F}(v)|^2\\
&=\inf_{u,v\in \mathcal{O}_K} \frac{1}{|b_d|^2}
|b_d N_{K/F}(u)-b_n N_{K/F}(v)|^2\in \frac{1}{|b_d|^2}  \mathcal{O}_K\cap \mathbb{R}^+
\end{align*}
with the infimum taken over all $(u,v)\not=(0,0)$.  Taking $c=|b_d|^2$ establishes the first part of the proposition.

Assume that $K$ is  quadratic imaginary (i.e. $F=\Q$ and  $a<0$). Then  it follows from Proposition~\ref{propA1} that   $\cO_K \cap \mathbb{R}^+=\mathbb{N}$ and that
$|b_d|^2=N_{K/\Q}(b_d)$ is a positive integer. Thus,   among all possible pairs
 $(b_n, b_d) \in \cO_K^2$ (with $b_d\not=0$) that satisfy $b=\frac{b_n}{b_d}$, 
we can choose a pair $(b_n,b_d)$ in such a way that $|b_d|^2$ is minimal. Let $d=|b_d|^2$. 
    
If $K$ is  not quadratic imaginary it follows from the Dirichlet Unit Theorem (Proposition~\ref{propA8}) that $\cO_K$ contains units $u$ such that $|u|^2$ is arbitrarily large or small,
so that  
$\delta(\mathcal{C}_{\mathcal{O}_K})$ can become arbitrarily small.
\end{proof}

\begin{remark} 
 It follows from the proposition that codes based on nonassociative quaternion algebras only satisfy the NVD property  if we assume that 
$K$ is a quadratic imaginary number field.

Assume that $K$ is  quadratic imaginary.   If we assume in addition that $\cO_K$ is a unique factorization domain (or, equivalently, a principal ideal domain; cf.  Appendix~\ref{AppA}), we can write $b$ as an irreducible fraction $b=b_n/b_d$ and $b_d$ will be unique up to multiplication by a unit. By the Dirichlet Unit Theorem the only units in $\cO_K$ are roots of unity, so that $|b_d|^2$ is unique. For information symbols $u$, $v$ taken from a QAM constellation we use the field $K=\Q(i)$ which is quadratic imaginary and whose ring of integers $\cO_K=\Z[i]$ is a unique factorization domain.
\end{remark}

\subsection{Information lossless encoding} A code $\cC$ is \emph{information lossless} if
it is obtained from information symbols in such a way that the energy needed to transmit them is the same as the energy needed to transmit the information symbols without encoding. By \cite[Prop. 3.5]{O-B-V} it suffices to construct the layers  of each codeword from the information symbols vector by applying a unitary matrix. This procedure is called \emph{cubic shaping} (cf. \cite[p.~3886]{O-R-B-V}), as it corresponds to an isometry transformation of the cubic lattice $\mathbb{Z}[i]^n$ in the case of information symbols taken from a QAM constellation. The term \emph{good shaping} is also used.

The energy needed to transmit  a complex number $z$ is determined by $|z|^2$. The energy needed to transmit a codeword $X=[x_{i,j}]\in \cC$ is determined by  its squared Frobenius norm $\Vert X\Vert^2 = \sum_{i,j} |x_{i,j}|^2$. 

Information lossless encoding of information symbols $c,d,e,f \in F$ into a codeword
\[ X=
\left [\begin {array}{cc}
x_0 &  b\sigma(x_1)\\
x_1 & \sigma(x_0)
\end {array}\right ]
\]
of $\cC$ can be done as follows. Let $\{u_0,u_1\}$ be an $F$-basis of $K$. Let
\[ G=
\left [\begin {array}{cc}
u_0 &  u_1\\
\sigma(u_0) & \sigma(u_1)
\end {array}\right ]
\]
be the matrix of the embeddings of  the basis. Let $x_0=cu_0+du_1$ and $x_1=eu_0+fu_1$ be elements of $K$ and consider the vectors $\mathbf{x}_0=(c,d)^T$, $\mathbf{x}_1=(e,f)^T$, each containing two information symbols.
Then
$$G \mathbf{x}_0=(x_0,\sigma(x_0))^T, \quad G \mathbf{x}_1=(x_1,\sigma(x_1))^T.$$
Let $\Gamma_1=I_2$ be the identity matrix and
\[ \Gamma_2=
\left [\begin {array}{cc}
0 &  b\\
1 & 0
\end {array}\right ].
\]
Then we can write $X$ as the sum of its layers,
$$X=\Gamma_1{\rm diag}(G\mathbf{x}_0)+\Gamma_2{\rm diag}(G\mathbf{x}_1)=\left [\begin {array}{cc}
x_0 &  0\\
0 & \sigma(x_0)
\end {array}\right ] +
\left [\begin {array}{cc}
0 &  b\sigma(x_1)\\
x_1 & 0
\end {array}\right ].$$
In order for the encoding to be information lossless, we want $\Gamma_2$ and $G$ to be unitary. Note that the matrix $\Gamma_2$ is unitary if and only if $|b|^2=1$.

In order to find a good unitary matrix $G$ (if one exists) and also to satisfy the NVD property, one usually restricts $x_0$ and $x_1$ to the ring of integers $\cO_K$ or an ideal $I$ of $\cO_K$ with ``good'' properties. See \cite[\S4.3-4.4]{O-B-V} for more details.
\bigskip

In the rest of this section we construct fully diverse $2\x 2$ codebooks $\lambda(A)$, based on nonassociative quaternion algebras $A$. From these we construct codebooks $\cC$ that satisfy the NVD and/or cubic shaping properties in certain cases. The codebooks are all infinite. When restricting entries of the codewords to the ring of integers $\cO_K$, we indicate this by writing $\cC_{\cO_K}$.

\subsection{Nonassociative Alamouti codes}

Let $F=\mathbb{R}$, $K=\mathbb{C}$ and let $\sigma=\can$ be  complex
conjugation. Recall \cite{Al} that the Alamouti Code is obtained from the quaternion division algebra $\mathbb{H}=(-1,-1)_\mathbb{R}$ over
$\mathbb{R}$ and yields codewords of the form
\[\left [\begin {array}{cc}
c+id &  -e+if\\
e+if & c-id
\end {array}\right ],
\]
with  $c+id,e+if$ the information symbols  ($c,d,e,f\in \mathbb{R}$). Used with QAM symbols it achieves the
diversity-multiplexing gain trade-off (DMT) of a MISO channel with 2 transmit antennas and 1 receive antenna.

\begin{example} \label{ex7.1} (i)
 Let $A=\Cay(\mathbb{C},i)$.  Then
\[\lambda(A)= \Biggl\{\left [\begin {array}{cc}
c+id &  f+ie\\
e+if & c-id
\end {array}\right ]\,\Bigg|\, c+id,e+if\in \mathbb{C} \Biggr\}.
\]
The codebook obtained from $\Cay(\mathbb{C},i)$ closely resembles the Alamouti Code.

Consider also  $B=\Cay(\mathbb{C},-i)$.  Then
\[\lambda(B)= \Biggl\{\left [\begin {array}{cc}
c+id &  -(f+ie)\\
e+if & c-id
\end {array}\right ]\,\Bigg|\, c+id,e+if\in \mathbb{C} \Biggr\}
\]
is another Alamouti-like code. Note that the algebras $A$ and $B$ are isomorphic by Remark~\ref{rem4.3}. 

Let us now consider QAM symbols only. 

(ii) Let  $A=\Cay(\mathbb{Q}(i),i)$ over $\mathbb{Q}$ and restrict the entries in $\lambda(A)$ to $\Z[i]$.
Then
\[\lambda(A)_{\Z[i]}= \Biggl\{\left [\begin {array}{cc}
c+id &  f+ie\\
e+if & c-id
\end {array}\right ]\,\Bigg|\, c+id,e+if\in \mathbb{Z}[i] \Biggr\}.
\]
From $\lambda(A)_{\Z[i]}$ we obtain a code $\cC_A$ with good shaping as follows: $\{1,i\}$ is a $\mathbb{Z}$-basis of $\mathbb{Z}[i]$ and
\[ \frac{1}{\sqrt{2}}G=\frac{1}{\sqrt{2}}
\left [\begin {array}{cc}
1 & i\\
1 & -i
\end {array}\right ]
\]
is a unitary matrix. So a codeword of $\cC_A$ is given by
\[\frac{1}{\sqrt{2}} \left [\begin {array}{cc}
c+id &  f+ie\\
e+if & c-id
\end {array}\right ],
\]
$c,d,e,f\in\mathbb{Z}$.

(iii) Similarly, $B=\Cay(\mathbb{Q}(i),-i)$ yields
\[\lambda(B)_{\Z[i]}= \Biggl\{\left [\begin {array}{cc}
c+id &  -(f+ie)\\
e+if & c-id
\end {array}\right ]\,\Bigg|\, c+id,e+if\in \mathbb{Z}[i] \Biggr\},
\]
resulting in a shaped code  $\cC_B$ with codewords 
\[\frac{1}{\sqrt{2}} \left [\begin {array}{cc}
c+id &  -(f+ie)\\
e+if & c-id
\end {array}\right ],
\]
$c,d,e,f\in\mathbb{Z}$.

In examples  (ii) and (iii), $F=\mathbb{Q}$, $K=\mathbb{Q}(i)$ is a quadratic imaginary number field, $b=\pm i$ and $\cO_K=\mathbb{Z}[i]$ is a principal ideal domain.
Hence, before shaping, the minimum determinant of each code is bounded below by $1$
by Proposition~\ref{prop6.2}. Thus the minimum determinant of both shaped codes is lower bounded by the constant~${1}/{2}$.

To summarize: the codes in (ii) and (iii) are fully diverse, satisfy the NVD property, have good shaping and clearly also satisfy the  uniform average transmitted energy per antenna property. Used for a $2\times 2$ MIMO channel
they are only half-rate though, since 4 transmitted signals are used to transmit 2 QAM  information symbols.
\end{example}

\subsection{Nonassociative Golden Codes}\label{S7.2}

  Let $F=\mathbb{Q}(i)$ and let $K= \mathbb{Q}(i)(\sqrt{5})$.  Then $\mathcal{O}_K=\mathbb{Z}[i][\frac{1+\sqrt{5}}{2}]$.
The Golden Code \cite{B-R-V} uses the maximal $\mathbb{Z}[i]$-order  in the quaternion division algebra
$(5,i)_{\mathbb{Q}(i)}=\Cay\bigl( \mathbb{Q}(i)(\sqrt{5}),i\bigr)$ which can be described
by the Cayley-Dickson doubling $\Cay\Bigl(\mathbb{Z}[i][\frac{1+\sqrt{5}}{2}],i\Bigr)$, defined in the obvious way. (Note that associative quaternion algebras are precisely the cyclic algebras of dimension $4$.) 
After
 shaping by
$\frac{1}{\sqrt{5}}$, a codeword of $\cC$ is thus of the form
\[ X=\frac{1}{\sqrt{5}}
\left [\begin {array}{cc}
c+d\theta &  e+f\theta\\
i(e+f\sigma(\theta) ) & c+d\sigma({\theta})
\end {array}\right ]
\]
with $\theta=\frac{1+\sqrt{5}}{2}$ the golden number, $\sigma:\mathbb{Q}(i)(\sqrt{5})\to \mathbb{Q}(i)(\sqrt{5})$,
$\sigma(i)= i$, $\sigma(\sqrt{5})= -\sqrt{5}$ and $c,d,e,f\in \mathbb{Z}[i]$.
To obtain an energy-efficient code, the entries in the codewords are then restricted to elements in the principal
ideal $I$ in $\mathcal{O}_K$ of norm 5 in $\mathbb{Q}$ which is generated by $\alpha=1+i-i\theta$. So, finally the Golden Code is
given by the codewords
\[ X=\frac{1}{\sqrt{5}}
\left [\begin {array}{cc}
\alpha(c+d\theta) &  \alpha(e+f\theta)\\
i\sigma(\alpha)(e+f\sigma(\theta)) & \sigma(\alpha) \sigma(c+d\sigma(\theta))
\end {array}\right ]
\]
with $c,d,e,f\in \mathbb{Z}[i]$ and has minimum determinant $1/5$. We refer to \cite{B-R-V} for the details.

\begin{example} \label{ex7.3}
Consider the nonassociative quaternion division algebra
$$A=\Cay\Bigl( \mathbb{Q}(i)(\sqrt{5}),\frac{i+\sqrt{5}}{i-\sqrt{5}}\Bigr)$$
 over
$\mathbb{Q}(i)$, where $|\frac{i+\sqrt{5}}{i-\sqrt{5}}|^2=1$ guarantees that $\Gamma_2$ is unitary.

The codebook based on $A$ is
\[\lambda(A)=\Biggl\{\left [\begin {array}{cc}
c+d\theta &  e+f\theta\\
\frac{i+\sqrt{5}}{i-\sqrt{5}}(e+f\sigma(\theta) ) & c+d\sigma(\theta)
\end {array}\right ]\,\Bigg|\,c,d,e,f\in \mathbb{Q}(i) \Biggr\}.
 \]
(Compared to the general code construction in Lemma~\ref{lem6.1}, we are transposing the matrices here in order  to better compare them  with the Golden Code matrices above. This does not influence the  behaviour of the code.) The code has full diversity and uniform average transmitted energy  per antenna. Now
$\{1,\theta\}$ is a $\mathbb{Q}(i)$-basis of $\mathbb{Q}(i)(\sqrt{5})$, but
\[ G=
\left [\begin {array}{cc}
1 &  \theta\\
1 & \sigma(\theta)
\end {array}\right ]
\]
is not a unitary matrix, so we have an energy shaping loss.  To obtain an energy-efficient code,  restrict
the entries in the codeword again to elements in the principal
ideal $I$ in $\mathcal{O}_K$ generated by $\alpha=1+i-i\theta$. Then a nonassociative Golden Code is
given by the codewords
\[ X=\frac{1}{\sqrt{5}}
\left [\begin {array}{cc}
\alpha(c+d\theta) &  \alpha(e+f\theta)\\
\frac{i+\sqrt{5}}{i-\sqrt{5}}\sigma(\alpha)(e+f\sigma(\theta)) & \sigma(\alpha) \sigma(c+d\sigma(\theta))
\end {array}\right ]
\]
with $c,d,e,f\in \mathbb{Z}[i]$. The choice of the ideal $I$ is optimal here for the exact same reasons as the ones
given in \cite[p.~1433]{B-R-V} and yields good shaping. The code is also
 full rate, fully diverse and has uniform average transmitted energy per antenna.
\end{example}

With the same arguments codes can be constructed using  any of the infinitely many scalars $b\in  \mathbb{Q}(i)(\sqrt{5})\setminus \Q(i)$ with $|b|^2=1$. All these codes, however, have vanishing determinant by Proposition~\ref{prop6.2}. We give another example:

\begin{example} \label{ex7.4}
Let $b=\frac{2i+\sqrt{5}}{3}$ and consider
the nonassociative quaternion algebra
$$\Cay\Bigl( \mathbb{Q}(i)(\sqrt{5}),\frac{2i+\sqrt{5}}{3}\Bigr)$$
 over $\mathbb{Q}(i)$. Again $|\frac{2i+\sqrt{5}}{3}|^2=1$ and we obtain another
  code which has full diversity
and uniform average transmitted energy per antenna. To obtain an energy-efficient code, we restrict
the entries in the codeword again to elements in the principal
ideal $I$ in $\mathcal{O}_K$ generated by $\alpha=1+i-i\theta$. Then another nonassociative Golden Code with good shaping is
given by the codewords
\[ X=\frac{1}{\sqrt{5}}
\left [\begin {array}{cc}
\alpha(c+d\theta) &  \alpha(e+f\theta)\\
\frac{2i+\sqrt{5}}{3}\sigma(\alpha)(e+f\sigma(\theta)) & \sigma(\alpha) \sigma(c+d\sigma(\theta))
\end {array}\right ]
\]
with $c,d,e,f\in \mathbb{Z}[i]$.
\end{example}

\subsection{Optimality of the Golden Code} 
Oggier \cite{Og} shows that the Golden Code is optimal inside the class of cyclic algebra based $2\times 2$ codes built over fields $K=\Q(i)(\sqrt{d})$    in the following sense: the minimum determinant of such codes is inversely proportional to $|d_{K/\Q(i)}|$, where $d_{K/\Q(i)}$ denotes the relative discriminant of $K/\Q(i)$. 
For the Golden Code $|d_{K/\Q(i)}|=5$. 
While it is possible to consider fields $K=\Q(i)(\sqrt{d})$ with $|d_{K/\Q(i)}|<5$, Oggier shows that  the resulting codes are no longer fully diverse \cite[III]{Og}. 

This problem does not occur in the nonassociative case by Theorem~\ref{thm4}. It is possible to construct fully diverse nonassociative codes over fields $K=\Q(i)(\sqrt{d})$ with $|d_{K/\Q(i)}|<5$,  but by Proposition~\ref{prop6.2} these codes do not satisfy the NVD property.  In the examples below we will consider the cases $|d_{K/\mathbb{Q}(i)}|=4$ and $|d_{K/\mathbb{Q}(i)}|=3$. The case  $|d_{K/\mathbb{Q}(i)}|=2$ does not exist, cf. Proposition~\ref{propA5}.

\begin{example}\label{ex7.5}
Let  $K=\mathbb{Q}(i)(\sqrt{2})$.
Then $|d_{K/\mathbb{Q}(i)}|=4$ and $\sigma(\sqrt{2})=-\sqrt{2}$. Moreover, $K= \mathbb{Q}(i)(\zeta_8)$ where $\zeta_8=\frac{1+i}{\sqrt{2}}$
is an 8th root of unity and $\sigma(\zeta_8)=-\zeta_8$.  We have that $\{1,\zeta_8\}$ is a $\mathbb{Z}[i]$-basis for the ring of integers
$\mathcal{O}_K=\mathbb{Z}[i][\zeta_8]$.

 Consider the nonassociative quaternion division algebra
$A=\Cay( \mathbb{Q}(\zeta_8),\zeta_8)$
over $ \mathbb{Q}(i)$.
The choice of $b=\zeta_8$  guarantees that $\Gamma_2$ is unitary, since $|\zeta_8|^2=1$.
We obtain the codebook
\begin{align*}
\lambda(A)&=\Biggl\{
\left [\begin {array}{cc}
x_0 &  \zeta_8\sigma(x_1)\\
x_1 & \sigma(x_0)
\end {array}\right ]\, \Bigg|\, x_0, x_1\in \mathbb{Q}(i)(\zeta_8) \Biggr\} \\
&=\Biggl\{
\left [\begin {array}{cc}
u_0+\zeta_8 w_0 &  \zeta_8(u_1-\zeta_8 w_1)\\
u_1+\zeta_8 w_1 & u_0-\zeta_8 w_0
\end {array}\right ] \,\Bigg|\, u_0, u_1,w_0, w_1\in \mathbb{Q}(i) \Biggr\}.
\end{align*}
Now $\{1,\zeta_8\}$ is a $\mathbb{Q}(i)$-basis of $\mathbb{Q}(i)(\zeta_8)$ and
\[ \frac{1}{\sqrt{2}}G=  \frac{1}{\sqrt{2}}
\left [\begin {array}{cc}
1 &  \zeta_8\\
1 & \sigma(\zeta_8)
\end {array}\right ]
=  \frac{1}{\sqrt{2}}
\left [\begin {array}{cc}
1 &  \zeta_8\\
1 & -\zeta_8
\end {array}\right ]
\]
is a unitary matrix. So after multiplying the matrices in the codebook by $ \frac{1}{\sqrt{2}}$ and restricting the information symbols to $\Z[i]$, we obtain a code that has good
shaping:
\[ \mathcal{C}=
\Biggl\{ \frac{1}{\sqrt{2}}
\left [\begin {array}{cc}
u_0+\zeta_8 w_0 &  \zeta_8(u_1-\zeta_8 w_1)\\
u_1+\zeta_8 w_1 & u_0-\zeta_8 w_0
\end {array}\right ] \,\Bigg|\, u_0, u_1,w_0, w_1\in \mathbb{Z}[i] \Biggr\}.
\]
The code $\cC$ is full rate, has full diversity and good shaping.
The factor $\zeta_8$ in the
first row of the codeword guarantees uniform average transmitted energy per antenna since $|\zeta_8|^2=1$.
This code does not satisfy the NVD property however by Proposition~\ref{prop6.2}.
\end{example}

\begin{example}
Let $K=\mathbb{Q}(i)(\sqrt{3})$. Then $|d_{K/\mathbb{Q}(i)}|=3$ 
 and $\sigma(\sqrt{3})=-\sqrt{3}$.
Moreover, $K=\mathbb{Q}(i)(\zeta_3)$ where $\zeta_3=e^{2\pi i /3}=\frac{-1+i\sqrt{3}}{2}$ is a third root of unity. We have $\sigma(\zeta_3)=\frac{-1-i\sqrt{3}}{2}=\overline{\zeta_3}$.
We know that $\{1,\zeta_3\}$ is a $\mathbb{Z}[i]$-basis for the ring of integers
$\mathcal{O}_K=\mathbb{Z}[i][\zeta_3]$.

Consider the nonassociative quaternion division algebra
$A=\Cay(\mathbb{Q}(i)(\zeta_3),\zeta_3)$ over $\mathbb{Q}(i)$. We obtain the codebook
\[\lambda(A)=\Biggl\{\left [\begin {array}{cc}
u_0+\zeta_3 w_0 &  \zeta_3(u_1+\overline{\zeta_3} w_1)\\
u_1+\zeta_3 w_1& u_0+\overline{\zeta_3}w_0
\end {array}\right ]\,\Bigg|\,u_0, u_1,w_0, w_1\in \mathbb{Q}(i) \Biggr\}.
 \]
 This  time the matrix $G$ (up to scaling) is not unitary. Thus the energy required to send the linear combination of the information symbols
 on each layer is higher than the energy needed
to send the information symbols themselves and we would
 still have to optimize for energy efficiency. In addition the discriminant of this code is not bounded away from  zero by Proposition~\ref{prop6.2}. 
\end{example}

\section{$2\times 4$ Multiblock space-time codes from nonassociative quaternion algebras}\label{sec8}

Let $F$  be a number field, let $a\in F^\x$ and 
let $K=F(\sqrt{a})$ be a quadratic field extension of $F$ with non-trivial Galois automorphism $\sigma$
and norm $N_{K/F}(x)=x\sigma(x)$.  Let $b\in K\setminus F$, so that
 $A=\Cay(K,b)$ is a nonassociative quaternion division algebra.
 A $2\x 4$  \emph{multiblock space-time code} based on $A$ is a set of matrices of the form $Y=[X |\sigma(X)]$ where $X\in \lambda(A)$
 (see \cite{L-M-M} for a more general construction in the associative case).

This yields a codebook $\cC$ consisting of  matrices of the form
\[Y= [X |\sigma(X)]= \left [\begin {array}{cc|cc}
x_0 &  b\sigma(x_1) &    \sigma(x_0) &  \sigma(b) x_1\\
x_1 & \sigma(x_0) &  \sigma(x_1) &  x_0   \\
\end {array}\right ]
 \]
with $x_0, x_1\in K$. They have full rank since $X$ comes from the division algebra $A$. In this set-up, we want the code to satisfy a \emph{generalized} NVD property (see \cite{Lu}) which can be achieved by bounding
the \emph{generalized minimum determinant}
$$\delta_g(\mathcal{C})=\inf_{0\not=X \in \lambda(A)}|\det(X) \det(\sigma(X))|$$
away from zero. 

\begin{proposition} 
Let $F$ be a number field and let $K=F(\sqrt{a})$ for some nonzero square-free
$a\in \mathcal{O}_F$. Let $b\in K\setminus F$. 
Let $A=\Cay(K,b)$ and let 
\[\mathcal{C}_{\mathcal{O}_K}=\{[X|\sigma(X)] \mid X\in \lambda(A)_{\cO_K}\}\]
denote the $2\x 4$ multiblock code  obtained from $\cC$
by restricting the elements of $K$ to elements of $\mathcal{O}_K$.  
Then there exists a constant $c>0$ such that 
\[\delta_g(\mathcal{C}_{\mathcal{O}_K})\in \frac{1}{c} \cO_F\cap \mathbb{R}^+.\]
If $F=\mathbb{Q}$ or $F$ is  quadratic imaginary, then there exists an integer $d>0$ such that
\[\delta_g(\mathcal{C}_{\mathcal{O}_K})\geq \frac{1}{\sqrt{d}}\]
$($and so $\mathcal{C}_{\mathcal{O}_K}$ satisfies the generalized NVD property$)$, otherwise $\delta_g(\mathcal{C}_{\mathcal{O}_K})$ can become arbitrarily small.
\end{proposition}

\begin{proof} Write $b$ as a fraction $b=\frac{b_n}{b_d}$ with $b_n,b_d\in \cO_K$ (not necessarily unique) and $b_d\not=0$.
We have 
\begin{align*}
\delta_g(\mathcal{C}_{\mathcal{O}_K})&=\inf_{0\not=X \in \lambda(A)_{\mathcal{O}_K}}|\det(X) \det(\sigma(X))|\\
&=\inf_{0\not=X \in \lambda(A)_{\mathcal{O}_K}}|\det(X) \sigma(\det(X))|\\
&=\inf_{0\not=X \in\lambda(A)_{\mathcal{O}_K}}|N_{K/F}(\det(X))|\\
&=\mathop{\inf_{x_0,x_1\in \mathcal{O}_K}}_{(x_0,x_1)\not=(0,0)  }|N_{K/F}(N_{K/F}(x_0)-bN_{K/F}(x_1))|\\
&=\mathop{\inf_{x_0,x_1\in \mathcal{O}_K}}_{(x_0,x_1)\not=(0,0)  }  \frac{1} {|N_{K/F}(b_d)|}
|N_{K/F}(b_dN_{K/F}(x_0)-b_nN_{K/F}(x_1))|\\
& \in \frac{1}{|N_{K/F}(b_d)|} \cO_F \cap \R^+
\end{align*}
since $N_{K/F}(\mathcal{O}_K)\subset \mathcal{O}_F$. Taking $c= |N_{K/F}(b_d)|$ establishes the first part of the proposition.

Assume that $F=\Q$. Then $|N_{K/F}(b_d)|$ is a positive integer. Thus,   among all possible pairs
 $(b_n, b_d) \in \cO_K^2$ (with $b_d\not=0$) that satisfy $b=\frac{b_n}{b_d}$, 
we can choose a pair $(b_n,b_d)$ in such a way  that $|N_{K/F}(b_d)|$ is minimal . Furthermore, $\cO_F\cap \R^+ =\N$. We let $d= |N_{K/F}(b_d)|^2$ in this case.

Next assume that $F$ is  quadratic imaginary (i.e. $F=\Q(\sqrt{m})$ and  $m<0$). Then it follows
from Proposition~\ref{propA1} that $\cO_F \cap \mathbb{R}^+=\mathbb{N}$ and that $|N_{K/F}(b_d)|^2=N_{F/\Q}(N_{K/F}(b_d))$ is a positive integer. 
Thus,   among all possible pairs
 $(b_n, b_d) \in \cO_K^2$ (with $b_d\not=0$) that satisfy $b=\frac{b_n}{b_d}$, 
we can choose a pair $(b_n,b_d)$ in such a way that
 $|N_{K/F}(b_d)|^2$ is minimal. Let $d=|N_{K/F}(b_d)|^2$.

If $F$ is  not $\Q$ or not quadratic imaginary it follows from the Dirichlet Unit Theorem (Proposition~\ref{propA8}) that $\cO_F$ contains units $u$ such that $|u|^2$ is arbitrarily large or small,  so that  
$\delta_g(\mathcal{C}_{\mathcal{O}_K})$ can become arbitrarily small.
\end{proof}

\begin{remark}  
If $F=\mathbb{Q}$ or $F$ is  quadratic imaginary, the generalized minimum determinant  of $\cC_{\cO_K}$ is lower bounded by a positive constant and the generalized NVD property is satisfied. As a consequence the code will achieve the diversity-multiplexing gain trade-off, as explained in \cite[p.~5232]{L-M-M}.

 If we assume in addition that $\cO_K$ is a unique factorization domain (or, equivalently, a principal ideal domain; cf.  Appendix~\ref{AppA}), we can write $b$ as an irreducible fraction $b=b_n/b_d$ and $b_d$ will be unique up to multiplication by a unit. By the Dirichlet Unit Theorem the only units in $\cO_K$ are roots of unity, so that $ |N_{K/F}(b_d)|^2$ is unique.

For QAM constellations we take $F=\Q(i)$, so that $K=\Q(i)(\sqrt{m})$ for some square-free non-zero integer $m$. Proposition~\ref{propA6} lists the fields $K$ whose ring of integers $\cO_K$ is a unique factorization domain.
\end{remark}

In the setting of multiblock space-time codes it is again  natural to ask that 
$|b|^2=1$, cf. \cite[p.~5232]{L-M-M}.

\begin{example} Let $F=\Q(i)$ and
$K=\mathbb{Q}(i)(\sqrt{5})$. Let $\theta=\frac{1+\sqrt{5}}{2}$ be the golden number, $\sigma:\mathbb{Q}(i)(\sqrt{5})\to \mathbb{Q}(i)(\sqrt{5})$,
$\sigma(i)= i$, $\sigma(\sqrt{5})= -\sqrt{5}$, $b=\frac{i+\sqrt{5}}{i-\sqrt{5}}$ and
$A=\Cay( \mathbb{Q}(i)(\sqrt{5}),b)$
over $\mathbb{Q}(i)$ as in Example~\ref{ex7.3}. In order to obtain an energy-efficient code, we restrict
the entries in the codewords  to elements in the principal
ideal $I$ in $\mathcal{O}_K$,  generated by $\alpha=1+i-i\theta$. Then
\[ X=\frac{1}{\sqrt{5}}
\left [\begin {array}{cc}
\alpha(c+d\theta) & b\sigma(\alpha)(e+f\sigma(\theta)) \\
\alpha(e+f\theta) & \sigma(\alpha) (c+d\sigma(\theta))
\end {array}\right ]
\]
with $c,d,e,f\in \mathbb{Z}[i]$ and the code $\cC_{\cO_K}$ consists of  block matrices of the form
\[Y=\frac{1}{\sqrt{5}} \left [\begin {array}{cc|cc}
\alpha(c+d\theta) &  b\sigma(\alpha)(e+f\sigma(\theta)) &  \sigma(\alpha)(c+d\theta) &   \sigma(b)\alpha(e+f \theta)    \\
\alpha(e+f\theta) &  \sigma(\alpha) (c+d\sigma(\theta))   &   \sigma(\alpha)  (e+f\sigma(\theta)) &   \alpha (c+d \theta)\\
\end {array}\right ].
 \]
Note that
$$\delta_g(\mathcal{C}_{\mathcal{O}_K})\geq \frac{1}{|N_{K/F}(\sqrt{5}(i-\sqrt{5}   )  )|}=\frac{1}{30},$$
 guaranteeing that the code satisfies the generalized NVD property.
\end{example}

\begin{example} Replacing $b$ by $\frac{2i+\sqrt{5}}{3}$ (cf. Example~\ref{ex7.4})
in the previous example results in a code $\cC_{\cO_K}$ such that
$$\delta_g(\mathcal{C}_{\mathcal{O}_K})\geq \frac{1}{|N_{K/F}(3\sqrt{5}  )|}=\frac{1}{45},$$
 guaranteeing that the code satisfies the generalized NVD property.
\end{example}

\begin{example} Let  $F=\Q(i)$ and 
$K=\mathbb{Q}(i)(\sqrt{2})= \mathbb{Q}(i)(\zeta_8)$ where $\zeta_8=\frac{1+i}{\sqrt{2}}$
is an 8th root of unity as in Example~\ref{ex7.5}. Let $A=\Cay( \mathbb{Q}(\zeta_8),\zeta_8)$. Then
\[X=
\frac{1}{\sqrt{2}}
\left [\begin {array}{cc}
u_0+\zeta_8 w_0 &  \zeta_8(u_1-\zeta_8 w_1)\\
u_1+\zeta_8 w_1 & u_0-\zeta_8 w_0
\end {array}\right ]\]
with $ u_0, u_1,w_0,w_1\in \mathbb{Z}[i] $ and the code $\cC_{\cO_K}$ consists of block matrices of the form
\[Y=\frac{1}{\sqrt{2}} \left [\begin {array}{cc|cc}
u_0+\zeta_8 w_0 &  \zeta_8(u_1-\zeta_8 w_1) &  u_0-\zeta_8 w_0 &  -\zeta_8 
(u_1+\zeta_8 w_1)\\
u_1+\zeta_8 w_1 & u_0-\zeta_8 w_0   &    u_1-\zeta_8w_1 &  u_0+\zeta_8 w_0\\
\end {array}\right ].
 \]
We have
$$\delta_g(\mathcal{C}_{\mathcal{O}_K})\geq \frac{1}{|N_{K/F}((\sqrt{2})^2  )|}=\frac{1}{4},$$
 guaranteeing that the code satisfies the generalized NVD property. 
\end{example}

\section{$4\times 4$ Codebooks from nonassociative quaternion division algebras}
\label{sec9}

Let $\mathbb{H}=(-1,-1)_\mathbb{R}$ be Hamilton's quaternion division algebra.
Its left regular representation with respect to the  basis $\{1,i,j,-ij\}$  consists of matrices of the form
\[ \left [\begin {array}{crrr}
x_0 &  -x_1 & -x_2 & -x_3 \\
x_1 & x_0  &  x_3 & - x_2 \\
x_2 & - x_3 & x_0 & x_1 \\
x_3 & x_2  & -x_1 & x_0 \\
\end {array}\right ]
 \]
with $x_\ell\in\mathbb{R}$, $\ell=0,\ldots, 3$ (cf. \cite[Example~8]{S-R-S}). This is exactly the
four-dimensional real
orthogonal design from \cite[Section III-A]{T-J-C1}.

Let us look at the left regular representation of a nonassociative
quaternion algebra $A$, this time over its base field rather than over a maximal subfield.

\subsection{Fully diverse codebook construction}

Let $F$ be a number field and let
$K=F(\sqrt{a})=F(i)$ with $i^2=a\in F^\x$ be a quadratic field extension with non-trivial Galois automorphism $\sigma:\sqrt{a}\mapsto -\sqrt{a}$.
Let $A=\Cay(K,b)$ be a nonassociative quaternion division algebra over $F$ with
$b=p+qi\in K\setminus F$, so $p,q\in F$ with $q\not=0$.
For the basis  $\{1,i,j,-ij\}$ of $A$ over $F$ the matrix representation of left multiplication with $x=x_0+x_1i+x_2j-x_3ij$
 yields the fully diverse  $4\times 4$ space-time
block code
\[ \mathcal{C}=\lambda(A)=\left.\left\{\left [\begin {array}{cccc}
x_0 &  ax_1 & px_2 -aqx_3 & aqx_2-apx_3 \\
x_1 & x_0  & qx_2-p x_3 & p x_2-aqx_3 \\
x_2 & a x_3 & x_0 & -ax_1 \\
x_3 & x_2  & -x_1 & x_0 \\
\end {array}\right ]\,\right|\, x_0,x_1,x_2,x_3\in F \right\}.
\]

 \begin{example}\label{ex9.1} Let $i^2=-1$.
 The $\mathbb{R}$-algebra $A=\Cay(\mathbb{C},i)$  yields the fully diverse $4\times 4$ space-time block code
\[\left.\left\{\left [\begin {array}{crrr}
x_0 &  -x_1 & x_3 & -x_2 \\
x_1 & x_0  & x_2 & x_3 \\
x_2 & - x_3 & x_0 & x_1 \\
x_3 & x_2  & -x_1 & x_0 \\
\end {array}\right ]\,\right|\, x_0,x_1,x_2,x_3\in \mathbb{R}\right\}.
 \]
Its matrices are not orthogonal, but their first two column vectors and, respectively,
their last two, are orthogonal to each other.
 \end{example}

\subsection{Non-vanishing determinant} Let
 $X\in \mathcal{C}$. Then
\[\det(X)=
[(x_0^2-ax_1^2)-p(x_2^2-ax_3^2)]^2-aq^2(x_2^2-ax_3^2)^2 \in F.
\]
Since the codebook is based on a division algebra, its minimum determinant equals
$$\delta(\mathcal{C})=\inf_{0\not=X\in \mathcal{C}}|\det(X)|^2$$
and is non-zero. 
If the information symbols $x_0,x_1,x_2,x_3$ belong to a finite constellation in $F$,
 then $\delta(\mathcal{C})$ is bounded below by a constant which depends on the constellation size. If the constellation size increases, $\delta(\mathcal{C})$ can get arbitrarily close to zero. By restricting the entries in $\cC$ to the ring of integers $\cO_F$ we obtain for certain number fields $F$ infinite codes that satisfy the NVD property:

\begin{proposition}
 Let $F$ be a number field and let $K=F(\sqrt{a})$ for some nonzero square-free
$a\in \mathcal{O}_F$. Let $b=p+q\sqrt{a}\in K\setminus F$ with $p,q\in F$
$($so that $q\not=0$$)$. 
Let $\cC=\lambda(\Cay(K,b))$ and let $\mathcal{C}_{\mathcal{O}_F}$ denote the code  obtained from $\cC$
by restricting the elements of $F$ to elements of $\mathcal{O}_F$.  
Then there exists a constant $c>0$ such that 
\[\delta(\mathcal{C}_{\mathcal{O}_F})\in \frac{1}{c} \cO_F\cap \mathbb{R}^+.\]
If $F=\Q$ or $F$ is quadratic imaginary, then there exists an integer $d>0$ such that
\[\delta(\mathcal{C}_{\mathcal{O}_F})\geq \frac{1}{d}\]
$($and so $\mathcal{C}_{\mathcal{O}_F}$ satisfies the NVD property$)$, otherwise $\delta(\mathcal{C}_{\mathcal{O}_F})$ can become arbitrarily small.
\end{proposition}

\begin{proof} Write $p=\frac{p_n}{p_d}$, $q=\frac{q_n}{q_d}$  
with $p_n,p_d, q_n,q_d\in \mathcal{O}_F$ (not necessarily unique) and $p_d\not=0$, $q_d\not=0$.
 An easy calculation confirms that
$$
\delta(\mathcal{C}_{\mathcal{O}_F})\in \frac{1}{|q_d|^4} \mathcal{O}_F\cap \mathbb{R}^+ \text{ if } p=0,\quad
\delta(\mathcal{C}_{\mathcal{O}_F})\in \frac{1}{|p_dq_d|^4} \mathcal{O}_F\cap \mathbb{R}^+ \text{ if } p\not=0.$$
Letting $c=|q_d|^4$ if $p=0$ and  $c=|p_dq_d|^4$ if $p\not=0$ establishes the first part of the proposition.

Assume for the sake of argument that $p=0$. The case $p\not=0$ can be settled in a similar manner.

Assume that $F=\Q$.  Then $|q_d|^4$  is a positive integer. 
Thus,   among all possible pairs
 $(q_n, q_d) \in \cO_F^2=\Z^2$ (with $q_d\not=0$) that satisfy $q=\frac{q_n}{q_d}$, 
we can choose a pair $(q_n,q_d)$ in such a way that
 $|q_d|^4$ is minimal. Furthermore, $\cO_F\cap \R^+ =\N$.
We let $d=|q_d|^4$ in this case.

Next assume that  $F$ is  quadratic imaginary (i.e. $F=\Q(\sqrt{m})$ and  $m<0$). Then it follows from Proposition~\ref{propA1} that 
$\cO_F \cap \mathbb{R}^+=\mathbb{N}$ and that
$|q_d|^4=N_{F/\Q}(q_d)^2$ is a positive integer.
Thus,   among all possible pairs
 $(q_n, q_d) \in \cO_F^2$ (with $q_d\not=0$) that satisfy $q=\frac{q_n}{q_d}$, 
we can choose a pair $(q_n,q_d)$ in such a way that
$|q_d|^4$ is minimal. Let $d=|q_d|^4$. 

If $F$ is  not $\Q$ or not quadratic imaginary it follows from the Dirichlet Unit Theorem (Proposition~\ref{propA8}) that $\cO_F$ contains units $u$ such that $|u|^2$ is arbitrarily large or small,  so that  
$\delta(\mathcal{C}_{\mathcal{O}_F})$ can become arbitrarily small.
\end{proof}

\begin{remark}  
If $F=\mathbb{Q}$ or $F$ is  quadratic imaginary, the minimum determinant  of $\cC_{\cO_F}$ is lower bounded by a positive constant  
and the NVD property is satisfied. 

If we assume in addition that $\cO_F$ is a unique factorization domain (or, equivalently, a principal ideal domain; cf.  Appendix~\ref{AppA}), we can write 
$p$ and $q$ as  irreducible fractions $p=\frac{p_n}{p_d}$, $q=\frac{q_n}{q_d}$ 
and $p_d$, $q_d$ will be unique up to multiplication by a unit. By the Dirichlet Unit Theorem the only units in $\cO_F$ are roots of unity, so that $|q_d|^4$, resp. 
$|p_dq_d|^4$, is unique.
\end{remark}

 \begin{example} \label{ex9.2}
The $\mathbb{Q}$-algebra $A=\Cay(\mathbb{Q}(i),i)$ yields the fully diverse $4\times 4$
space-time block code
\[ \mathcal{C}_\mathbb{Z}=\left.\left\{\left [\begin {array}{cccc}
x_0 &  -x_1 & x_3 & -x_2 \\
x_1 & x_0  & x_2 & x_3 \\
x_2 & - x_3 & x_0 & x_1 \\
x_3 & x_2  & -x_1 & x_0 \\
\end {array}\right ]\,\right|\, x_0, x_1, x_2, x_3\in \mathbb{Z}\right\}.
 \]
We obtain
$$\delta(\mathcal{C}_{\mathbb{Z}})=
\mathop{\inf_{x_0, x_1, x_2, x_3\in \mathbb{Z}}}_{ 
(x_0, x_1, x_2, x_3)\not=(0,0,0,0)}
|(x_0^2+x_1^2)^2+(x_2^2+x_3^2)^2 |^2=1.$$
\end{example}

\begin{example}\label{ex9.2bis}
The $\mathbb{Q}$-algebra $A=\Cay(\mathbb{Q}(i),-i)$ yields the fully diverse $4\times 4$ space-time block code
\[ \mathcal{C}_\mathbb{Z}=\left.\left\{\left [\begin {array}{cccc}
x_0 &  -x_1 & -x_3 & x_2 \\
x_1 & x_0  & -x_2 & -x_3 \\
x_2 & - x_3 & x_0 & x_1 \\
x_3 & x_2  & -x_1 & x_0 \\
\end {array}\right ]\,\right|\, x_0, x_1, x_2, x_3\in \mathbb{Z}\right\},
 \]
 again with minimum determinant $1$.
 \end{example}

In the previous two examples the first two column vectors of any codeword  and, respectively, the last two, are orthogonal to each other.

 \begin{example} \label{ex9.3}
Consider the  $ \mathbb{Q}(i)$-algebra
$A=\Cay( \mathbb{Q}(i)(\zeta_8),\zeta_8)$  where
 $\zeta_8=\frac{1+i}{\sqrt{2}}$ is an 8th root of unity. Note that $a=\zeta_8^2=i$ and $b=\zeta_8$ (so $p=0$ and $q=1$).
 We obtain the fully diverse codebook
\[ \mathcal{C}_{\mathbb{Z}[i]}=\left.\left\{\left [\begin {array}{cccc}
x_0 &  ix_1 & -ix_3 & ix_2 \\
x_1 & x_0  &  x_2 &  -ix_3 \\
x_2 & i x_3 & x_0 & -ix_1 \\
x_3 & x_2  & -x_1 & x_0 \\
\end {array}\right ]\,\right|\, x_0, x_1, x_2, x_3\in \mathbb{Z}[i]\right\}
 \] 
 whose minimum determinant is
$$\delta(\mathcal{C}_{\mathbb{Z}[i]})=
\mathop{\inf_{x_0,x_1,x_2,x_3\in \mathbb{Z}[i]}}_{(x_0, x_1, x_2, x_3)\not=(0,0,0,0)}|(x_0^2-ix_1^2)^2-i(x_2^2-ix_3^2)^2|^2=1.$$

\end{example}

\subsection{Information lossless encoding}

For a (transposed) matrix
\[ X=
\left [\begin {array}{cccc}
x_0 &  x_1 & x_2 & x_3 \\
ax_1 & x_0  & a x_3 &  x_2 \\
 px_2 -aqx_3 & qx_2-p x_3  & x_0 & -x_1 \\
aqx_2-apx_3  &  p x_2-aqx_3  & -ax_1 & x_0 \\
\end {array}\right ]
\]
in $\cC$
we use the following encoding: let $I_4$ be the identity matrix and let
\[ \Gamma_1=
\left [\begin {array}{cccc}
0 & 1 & 0 & 0\\
a & 0 & 0 & 0 \\
0 & 0 & 0 & -1 \\
0 & 0 & -a & 0 \\
\end {array}\right],\
\Gamma_2=
\left [\begin {array}{cccc}
0 & 0 & 1 & 0\\
0 & 0 & 0 & 1 \\
p & q & 0 & 0 \\
aq & p & 0 & 0 \\
\end {array}\right ],\
\Gamma_3=
\left [\begin {array}{cccc}
0 & 0 & 0 & 1\\
0 & 0 & a & 0 \\
-aq & -p & 0 & 0 \\
-ap & -aq & 0 & 0 \\
\end {array}\right ].
\]
The codeword $X$ is encoded as
$$X=I_4 {\rm diag}(x_0)+\Gamma_1{\rm diag}(x_1)+\Gamma_2{\rm diag}(x_2) +
\Gamma_3{\rm diag}(x_3),$$
where, for $\ell=0,\ldots,3$,
\[ {\rm diag}(x_\ell)=
\left [\begin {array}{cccc}
x_\ell & 0 & 0 & 0\\
0 & x_\ell & 0 & 0 \\
0 & 0 & x_\ell & 0 \\
0 & 0 & 0 & x_\ell \\
\end {array}\right ].
\]
The matrix $\Gamma_3$ is unitary if and only if $|a|^2=1$, $ aq\bar{p}+p\bar{q}=0$ and $|p|^2+|q|^2=1$. The matrix
$\Gamma_2$ is unitary if and only if $|p|^2+|q|^2=1$, $|a|^2|q|^2+|p|^2=1$ and $ aq\bar{p}+p\bar{q}=0$.
The matrix $\Gamma_1$ is unitary if and only if $|a|^2=1$.

Thus,  $\cC$ is  information lossless  
if $|a|^2=|p|^2+|q|^2=1$ and $p\bar q+aq\bar p=0$.

It is not difficult to verify that the codes in Examples~\ref{ex9.2}, \ref{ex9.2bis} and \ref{ex9.3} 
are all information lossless.

\appendix

\section{Facts from Number Theory}\label{AppA}

In this appendix we collect some results from algebraic number theory for the convenience of the reader.

Let $K$ be a number field.  The ring of integers $\cO_K$ of $K$ is a Dedekind domain \cite[I(3.1)]{Neu}.

Let $d_K$ denote the discriminant of $K$.

\begin{proposition}[{\cite[p.15]{Neu}}] \label{propA1}
Let $m\not=0$  be a square-free integer  and let $K=\Q(\sqrt{m})$. Then
\[d_K=\begin{cases}
4m & \text{if } m\equiv 2,3 \bmod 4\\
m &\text{if } m\equiv 1 \bmod 4
\end{cases}.\]
An integral basis of $K$ is given by $\{1,\sqrt{m}\}$ in the first case, by $\{1, \frac{1}{2} (1+\sqrt{m})\}$ in the second case and by  $\{1, \frac{1}{2} (m+\sqrt{m})\}$ in both cases. Thus
\[\cO_K=\begin{cases}
\Z[\sqrt{m}] & \text{if } m\equiv 2,3 \bmod 4\\
\Z\bigl[\frac{1+\sqrt{m}}{2}\bigr] &\text{if } m\equiv 1 \bmod 4
\end{cases}.\]
\end{proposition}

Let $h_K$ denote the class number of $K$, then $h_K=1$ if and only if $\cO_K$ is a principal ideal domain \cite[I,\S6]{Neu} if and only if $\cO_K$ is a unique factorization domain (since $\cO_K$ is a Dedekind domain \cite[Prop. 3.18]{Mil}).

\begin{proposition}[{\cite[p. 48]{Mil}}] Let $m$ be a positive square-free integer and let $K=\Q(\sqrt{-m})$. Then $h_K=1$ if and only if $m\in\{1,2,3,7,11,19,43,67,163\}$.
\end{proposition}

For an extension of number fields $K/F$, let $d_{K/F}$ denote the relative discriminant of $K$ over $F$.

\begin{proposition}[{\cite[p. 443]{Hasse}}] \label{propA3}
Let $L\supset K \supset F$ be a chain of number fields, then
\[d_{L/F}=N_{K/F}(d_{L/K})d^n_{K/F},\]
where $n=[K:F]$.
\end{proposition}

Let $i=\sqrt{-1}$. We collect some useful facts about quadratic extensions of $\Q(i)$.

\begin{proposition}[{\cite[Satz 2.1]{Schmal}}]  \label{propA4}
Let $m\not=0$  be a square-free integer  and let $K=\Q(i)(\sqrt{m})$. Then
\[d_K=\begin{cases}
16m^2 & \text{if } m\equiv 1,3 \bmod 4\\
64m^2 &\text{if } m\equiv 2 \bmod 4
\end{cases}.\]
\end{proposition}

\begin{proposition}\label{propA5} 
Let $m\not=0$  be a square-free integer  and let $K=\Q(i)(\sqrt{m})$. Then there exists a relative integral basis of $K$ over $\Q(i)$. Furthermore,
\[|d_{K/\Q(i)}|=\frac{1}{4} \sqrt{d_K}.\]
In particular,
\[|d_{K/\Q(i)}|=\begin{cases}
|m| & \text{if } m\equiv 1,3 \bmod 4\\
2|m| &\text{if } m\equiv 2 \bmod 4
\end{cases}.\]
\end{proposition}

\begin{proof} The existence of a relative integral basis follows from \cite[Satz 4.2a]{Schmal}.  Consider the chain $K\supset \Q(i)\supset \Q$.
From Proposition~\ref{propA1} it follows that $d_{\Q(i)}=-4$. Thus,
$d_K=|d_{K/\Q(i)}|^2(-4)^2$ by Proposition~\ref{propA3}, which shows that $|d_{K/\Q(i)}|=\frac{1}{4} \sqrt{d_K}$. We conclude with Proposition~\ref{propA4}.
\end{proof}

\begin{proposition}[{\cite[pp. 915--916]{Yam}}] \label{propA6}
Let $m$ be a positive square-free integer and let 
$K=\Q(i)(\sqrt{m})$. Then $h_K=1$ if and only if $m\in\{2,3,5,7,11,13,19,37,43,67,163\}$.
\end{proposition}

\begin{remark}  Note that $\Q(i)(\sqrt{m})=\Q(i)(\sqrt{-m})$. Also note that $\Q(i)(\sqrt{2})=\Q(\zeta_8)$ and $\Q(i)(\sqrt{3})=\Q(\zeta_{12})$, where $\zeta_n$ denotes a primitive $n$-th root of unity.
\end{remark}

\begin{proposition}[Dirichlet's Unit Theorem {\cite[I(7.4)]{Neu}}]\label{propA8} 
Let $K$ be a number field with ring of integers $\cO_K$. Let $r$ be the number of real embeddings of $K$ and $s$ the number of pairs of complex conjugate embeddings of $K$. Let $\mu(K)$ denote the finite cyclic group of roots of unity that lie in $K$. The group of units of $\cO_K$  is the direct product of $\mu(K)$ and a free abelian group of rank $r+s-1$.
\end{proposition}

\section*{Acknowledgements} The authors wish to thank CIRM (Centro Internazionale per la Ricerca Matematica) and the Fondazione Bruno Kessler for financial support via the Research in Pairs programme and hospitality  during the period December 13--21, 2008. They warmly thank Augusto Micheletti for facilitating their visit.  The second author also wishes to thank  Sandro Mattarei and the University of Trento for their hospitality and support during the same period. The authors wish to thank 
Fr\'ed\'erique Oggier and
the referees for their constructive criticism of an earlier version of this paper.

\medskip

\medskip
 {\it E-mail address: }susanne.pumpluen@nottingham.ac.uk\\
 \indent{\it E-mail address: }thomas.unger@ucd.ie\\

\end{document}